\documentclass[aps,pra,showpacs,reprint,twocolumn,groupedaddress,floatfix]{revtex4-1}
\usepackage{graphicx}
\usepackage{amsmath}
\usepackage{latexsym}
\usepackage{amsfonts}
\usepackage{mathtools}
\usepackage{MnSymbol}
\usepackage{amsthm}
\usepackage{mathrsfs}
\usepackage{wasysym}
\usepackage{textcomp}
\usepackage{color} 
\usepackage{multirow}
\usepackage{pifont}
\usepackage[T1]{fontenc}
\usepackage{dcolumn}
\usepackage{bm}
\usepackage{balance}
\usepackage{color}
\usepackage{enumerate}

\theoremstyle{definition}
\newtheorem{definition}{Definition}[section]
\newtheorem{theorem}{Theorem}[section]

\allowdisplaybreaks

\begin{document}

\def\crta{\vrule height1.41ex depth-1.27ex width0.34em}
\def\dj{d\kern-0.36em\crta}
\def\Crta{\vrule height1ex depth-0.86ex width0.4em}
\def\Dj{D\kern-0.73em\Crta\kern0.33em}
\dimen0=\hsize \dimen1=\hsize \advance\dimen1 by 40pt

\title{Hypergraph Contextuality}

\author{Mladen Pavi\v ci\'c}
\email{mpavicic@irb.hr}
\homepage{http://www.irb.hr/users/mpavicic}
\affiliation{Center of Excellence for Advanced Materials and Sensors
  CEMS, Photonics and Quantum Optics
  Unit, Ru\dj er Bo\v skovi\'c Institute, Zagreb, Croatia.}

\date{November 10, 2019}


\begin{abstract}Quantum contextuality is a source of quantum computational
  power and a theoretical delimiter between classical and quantum
  structures. It has been substantiated by numerous experiments 
  and prompted generation of state independent contextual sets, i.e.,
  sets of quantum observables capable to reveal quantum contextuality
  for any quantum state of a given dimension. There are two major
  classes of state-independent contextual sets: the Kochen-Specker
  ones and the operator-based ones. In this paper, we present a third,
  hypergraph-based class of contextual sets. Hypergraph inequalities
  serve as a measure of contextuality. We limit ourselves to qutrits
  and obtain thousands of 3-dim contextual sets. The simplest of
  them involves only 5 quantum observables, thus enabling a
  straightforward implementation. They also enable establishing new
  entropic contextualities.
\end{abstract}

\maketitle

\section{\label{sec:intro}Introduction}

Recently, quantum contextuality found applications in quantum
communication \cite{cabello-dambrosio-11,nagata-05}, quantum
computation \cite{magic-14,bartlett-nature-14}, quantum
nonlocality \cite{kurz-cabello-14}, and lattice theory
\cite{bdm-ndm-mp-fresl-jmp-10,mp-7oa}. This has prompted 
experimental implementation with photons
\cite{simon-zeil00,michler-zeil-00,amselem-cabello-09,liu-09,d-ambrosio-cabello-13,ks-exp,ks-exp-03,lapkiewicz-11,zu-wang-duan-12,canas-cabello-8d-14,canas-cabello-14,zhan-sanders-16},
classical light
\cite{li-zeng-17,li-zeng-19,frustaglia-cabello-16,zhang-zhang-19},
neutrons \cite{h-rauch06,cabello-fillip-rauch-08,b-rauch-09},
trapped ions \cite{k-cabello-blatt-09},
solid state molecular nuclear spins  \cite{moussa-09},
and superconducting quantum systems \cite{jerger-16}.

Quantum contextuality, the aforementioned citations refer to,
precludes assignments of predetermined values to dense sets of
projection operators and in our approach we shall keep to this
feature of the considered contextual sets. Contextual theoretical
models and experimental tests involve additional subtle
issues, such as the possibility of classical noncontextual hidden
variable models that can reproduce quantum mechanical predictions
up to arbitrary precision \cite{barrett-kent-04} or a generalization
and redefinition of noncontextuality
\cite{kunjwal-spekkens-15,kunjwal-18-arxiv}.
These elaborations are outside of the scope of the present paper,
though, since it is primarily focused on contextuality which finds
applications within quantum computation vs.~noncontextuality which
is inherent in the current classical binary computation. That means
that we consider classical models with predetermined binary values,
that can be assigned to measurement outcomes of classical observables,
which underlie the latter computation, vs.~quantum models that do
not allow for such values and underlie quantum computation.
As for a direct relevance of our results for quantum computation
we point out that the hypergraph presented in Fig.~2 of
\cite{magic-14}, from which the contextual ``magic'' of quantum
computation has been derived, is a kind of hypergraph contextual
sets we present in this paper. However, the hypergraph is from a
4-dim Hilbert space, so, we will not elaborate on it in the paper. 

We give a pedestrian overview of our approach, methods, and results
as well as their background in the last few paragraphs of this
Introduction, describing the organization of the paper. 

A class of state-independent contextual (SIC)
\cite{beng-blan-cab-pla12} sets that have been elaborated on the
most in the literature are the Kochen-Specker (KS) sets
\cite{cabell-est-96a,pmmm05a,aravind10,waeg-aravind-jpa-11,mfwap-11,mp-nm-pka-mw-11,waeg-aravind-megill-pavicic-11,waegell-aravind-12,waeg-aravind-pra-13,waeg-aravind-fp-14,waeg-aravind-jpa-15,waeg-aravind-pla-17,pavicic-pra-17,pm-entropy18,pwma-19}. They boil down to a list of
$n$-dim vectors and their $n$-tuples of orthogonalities, such
that one cannot assign definite binary values to them.

Recently, different SIC sets has been designed and/or considered by
Yu and Oh \cite{yu-oh-12}, Bengtsson, Blanchfield, and Cabello
\cite{beng-blan-cab-pla12}, Xu, Chen, and Su
\cite{xu-chen-su-pla15}, Ramanathan and Horodecki \cite{ram-hor-14},
and Cabello, Kleinmann, and Budroni \cite{cabell-klein-budr-prl-14}.
They all make use of operators defined by vectors that define their
sets. You and Oh construct rather involved expression of state/vector
defined $3\times 3$ operators that eventually reduces to a multiple
of a unit operator while the other authors make use of projectors
whose expressions also reduce to a multiple of a unit operator.
Therefore we call their sets the {\em operator-based contextuality
sets} and assume that they form an {\em operator contextuality
class}. All the sets make use of a particular list of 3-dim vectors
and their orthogonal doublets and triplets such that a given
expression of definite binary variables has an upper bound which is
lower than the one of a corresponding quantum expression. The last
two Refs.~\cite{ram-hor-14,cabell-klein-budr-prl-14} also provide
us with the necessary and sufficient condition for being a SIC set
in any dimension.

The difference between the KS contextuality and the operator
contextuality is that KS statistics includes measured values of all
vectors from each $n$-tuple, while the statistics of measurements is
built on values obtained via operators defined by possibly less than
$n$ vectors from each $n$-tuple.

In this paper, we blend the two aforementioned contextualities
so as to arrive at a hypergraph one. We consider hypergraphs
with 3-dim vectors in which some of those vectors that belong to
only one triplet are dropped, as in the observable approach, and
generate smaller hypergraphs from them, such that one cannot assign
definite binary values to them, as in the KS approach. We call our
present approach the McKay-Megill-Pavi{\v c}i{\'c} hypergraph (MMPH)
approach. MMPH non-binary sets directly provide us with noncontextual
inequalities. On the other hand, via our algorithms and programs we
obtain thousands of smaller MMPH sets which can serve for various
applications as, e.g., to generate new entropic tests of
contextuality or new operator-based contextual sets. 

The smallest MMPH non-binary set we obtain is a pentagon with five
vectors (vertices) cyclically connected with 5 pairs of orthogonality
(edges). It corresponds to the pentagram from Ref.~\cite{klyachko-08},
implemented in
\cite{lapkiewicz-11,li-zeng-17,zhang-zhang-19}. The
difference is that the pentagram inequality is state dependent,
while the MMPH pentagon inequality is state independent.
More specifically, in Ref.~\cite{klyachko-08}, one obtains a
nonclassical inequality by means of projections of five pentagram
vectors at a chosen sixth vector directed along fivefold symmetry
axis of the pentagram. By our method, one gets a nonclassical
inequality between the maximum sum of possible assignments of 1,
representing classical measurement clicks, and the sum of
probabilities of obtaining quantum measurement clicks.

Entropic test of contextuality for pentagram/pentagon has been
formulated in Ref.~\cite{kurz-raman-12} following
Ref.~\cite{braun-cav-99}. It can be straightforwardly
reformulated for other MMPH non-binary sets we obtained.

The paper is organized as follows.

In Sec.~\ref{subsec:form}
we present the hypergraph formalism and define $n$-dim MMPH set
and $n$-dim MMPH binary and non-binary sets as well as {\em filled}
MMPH set. We explain how vertices and edges in an $n$-MMPH set
correspond to vectors and their orthogonalities, i.e., $m$-tuples
($2\le m\le n$) of mutually orthogonal vectors, respectively.

In Sec.~\ref{subsec:yuohks} we give the KS theorem and
a definition of a KS set and prove that a KS set is a special
non-binary set. In Def.~\ref{def:crit} we define a {\em critical}
KS set, i.e., the one which would stop being a KS set if we removed
any of its edges. Then we introduce known KS sets to compare
them with operator defined sets. In particular, we start with 
Conway-Cohen, Bub, Peres, and original Kochen-Specker's sets. We
show that the number of vectors they are characterised with in the
original papers and most of the subsequent ones as well as in books,
i.e., 31, 33, 33, and 117, respectively, are not critical.
That, actually, enables the whole approach presented in this
paper. We show that the aforementioned authors dropped the vectors 
that are contained in only one triplet. If we took all the stripped
vectors into account, i.e., if we formed filled sets, we would get
51, 49, 57, and 192 vectors, respectively. These sets are critical
and the majority of researchers assumed that their stripped versions
are critical too and so they did not try to use them as a source of
smaller non-classical 3-dim sets.

Next, we connect and compare KS sets with operator-based sets,
in particular YU-Oh's 13 vector set whose filled version has 25
vectors and 16 triplets---we denote it as 25-16. In
Fig.~\ref{fig:yu-oh-peres} we show Yu-Oh's 25-16 as a subgraph of
Peres' 57-40. In Fig.~\ref{fig:yu-oh} we show how 25-16 can be
stripped of vectors contained in only one triplet, so as to a
arrive at the original Yu-Oh's 13-16 set.
Eqs.~(\ref{eq:vec})-(\ref{eq:ine}) and their comments explain how
Yu and Oh defined their operators with the help of the 13 vectors
and how they used them to arrive, via Eq.~(\ref{eq:L}), at the
inequality defined by Eq.~(\ref{eq:ine}). We then used the operator
expression given by Eq.~(\ref{eq:L}) to test 50 sets smaller and
bigger than the 13-16 but did not obtain an analogous result.
Some of the sets are shown in Fig.~\ref{fig:d}.

In Sec.~\ref{subsec:masters} we give a historical background of
stripping the aforementioned vectors that are contained in only one
triplet and explain what was behind that ``incomplete triplets''
issue. Then we give MMPH strings of Conway-Kochen's 31-37, Bub's
33-36, Peres' 33-40, and Kochen-Specker's 117-118 non-critical
but still non-binary non-classical MMPH sets and take them as our
master sets from which we generate smaller non-binary critical
MMPH sets in the next section. However, we stress that any set we
obtain by stripping some other number of vertices contained in only
one edge from any one of the original four KS sets can serve
us as a master set. We give a Peres' 40-40 set as an example.

In Sec.~\ref{subsec:crit} we start with Def.~\ref{def:crit-nb}
of a critical MMPH non-binary set which differs from the one of
a critical KS set. If we strip more and more edges from a critical
KS set  we shall never come to a KS set again. This is not so 
with MMPH non-binary sets. MMPH non-binary critical sets might
properly contain smaller MMPH non-binary critical sets whose
number of edges is smaller than the original critical set for
at least 2 edges.

Via our algorithms and programs, we obtain thousands of critical
sets from our master sets, whose distributions are shown in
Fig.~\ref{fig:dis-strip}. We say that a collection of MMPH
non-binary subgraphs of an MMPH master form its class. 

Next we define measurements which can distinguish contextual from
non-contextual MMPH sets, i.e., non-binary from binary ones. 
Similarly as with operator-based contextual measurements,
dropped vertices are not considered, i.e., clicks obtained at their
corresponding out-ports are not taken into account when obtaining the
statistics of collected data. So, measurements of MMPH non-binary sets
are carried out as for KS sets with triplets, i.e., with the 1/3
probability of detection at each out-port, and via {\em calibrated}
detections of a particle or a photon at out-ports of a gate
representing a doublet with the 1/2 probability of getting a click
at each of the two considered ports, while ignoring the third one.
When a vertex shares a mixture of triplet and doublet edges the
probability of detection is $1/p$, where $1/3\le p\le 1/2$. We call
detections at all ports notwithstanding whether we include them in
our final statistics or not, {\em uncalibrated} detections---they
simply have 1/3 probability of detection at every port. 

To obtain contextual distinguishers of an MMPH set we consider the
sum of probabilities of getting clicks for all considered vertices
and call it a {\em quantum hypergraph index}. We distinguish a
calibrated quantum hypergraph index, which we denote as $HI_q$ and
an uncalibrated one, which we denote as $HI_{q-unc}$. On the other
hand, each MMPH set allows a maximal number of 1s assigned to
vertices so as to satisfy the two conditions from
Def.~\ref{df:mmphs}. We call the number {\em classical hypergraph
index} and denote it as $HI_c$. Our {\em weak} contextual
distinguisher is the inequality: $HI_q>HI_c$ and the {\em strong}
one is the inequality $HI_{q-unc}>HI_c$. Yu-Oh, Bub, Conway-Kochen,
and Peres' MMPH non-binary sets as well as others given in the
section, like, e.g., 13-10, satisfy both inequalities. 

We present several small critical MMPH sets in Figs.~\ref{fig:bub-c} 
and \ref{fig:c-k-c} and discuss their features. We also calculate
Yu-Oh's inequalities for several sets different from Yu-Oh's
13-16 set. None of the 50 tested sets satisfy the inequality.  

In Sec.~\ref{sec:disc} we discuss and reexamine the steps and
details of our approach. 


\section{Results}

We consider a set of quantum states represented by vectors in a 3-dim
Hilbert space ${\mathscr{H}}^3$ grouped in triplets of mutually
orthogonal vectors. We describe such a set by means of a hypergraph
which we call a {\em McKay-Megill-Pavi{\v c}i{\'c} hypergraph} (MMPH).
In it, vectors themselves are represented by vertices and mutually
orthogonal triplets of them by edges. However, an MMPH itself has
a definition which is independent of a possible representation of 
vertices by means of vectors. For instance, there are MMPHs without
a coordinatization, i.e., MMPHs for whose vertices vectors, one
could assign to, do not exist. Also, edges can contain less than 3
vertices, i.e., 2, and form doublets. When a coordinatization exist,
that does not mean that a doublet belongs to a 2-dim edge, but only
that we do not take an existing third vertex/vector into account.

\subsection{\label{subsec:form}Formalism}

Let us define the hypergraph formalism.

A hypergraph is a pair $v\text{-}e$ where $v$ is a set of elements
called vertices and $e$ is a set of non-empty subsets of $e$ called
edges. Edge is a set of vertices that are in some sense {\em related}
to each other, in our case {\em orthogonal} to each other.

The first definition of MMPH was given in \cite{pmmm05a-corr}
where we called them, not hypergraphs, but diagrams.
In \cite{pavicic-pra-17} we gave a definition of an $n$-dim
{\em MMP hypergraph} which required that each edge has at least
3 vertices and that edges that intersect each other in $n\text{-}2$
vertices contain at least $n$ vertices. The definition of
{\em MMPH} is slightly different.

\begin{definition}\label{def:mmp}An {\rm MMPH} is an $n$-dim
hypergraph in which
\begin{enumerate}
\item Every vertex belongs to at least one edge;
\item Every edge contains at least 2 vertices;
\item Edges that intersect each other in $m-2$ vertices contain at
  least $m$ vertices, where $2\le m\le n$.
\end{enumerate}
\end{definition}

Then, in \cite{pm-entropy18} we presented a hypergraph
reformulation of the Kochen-Specker theorem \cite{koch-speck}
from which we derive the following definition of an MMPH
non-binary set.

\begin{definition}\label{df:mmphs}
  $n$-dim {\rm MMPH non-binary set}, $n\ge 3$,
  is a hypergraph whose each edge contains at least two and at
  most $n$ vertices to which it is impossible to assign 1s and
  0s in such a way that
  \begin{enumerate}
\item No two vertices within any of its edges are both assigned
  the value 1;
\item In any of its edges, not all of the vertices are
  assigned the value 0.
\end{enumerate}
An {\rm MMPH} set to which it is possible to assign 1s and 0s so as to
  satisfy the above two conditions we call an {\rm MMPH binary set}.\\
  An {\rm MMPH} non-binary set with edges of mixed sizes to which
  vertices are added so as to make all edges of equal size each
  containing $n$ vertices is called {\rm filled MMPH set}.
\end{definition}

A coordinatization of an MMPH non-binary set means that the vertices
of its {\em filled} MMPH denote $n$-dim vectors in
${\mathscr{H}}^n$, $n\ge 3$ and that its edges represent orthogonal
$n$-tuples, containing vertices corresponding to those mutually
orthogonal vectors. Then the vertices of an MMPH set with edges
of mixed sizes inherit its coordinatization from the 
coordinatization of its filled set. In our present approach a
coordinatization is automatically assigned to each hypergraph by
the very procedure of its generation from master MMPHs as we shall
see below. 

In the real 3-dim Hilbert space edges form loops of order five
(pentagon) or higher as we proved in \cite{pmmm05a}. For complex
vectors our calculations always confirmed this result but we were
unable to find an exact proof. Loops of order two are
precluded by Def.~\ref{def:mmp}(3).

MMPH are encoded by means of printable ASCII characters organized
in a single string, and within it in edges, which are separated by
commas; each string ends with a period. Vertices are
denoted by one of the following characters: {{\tt 1 2 \dots\ 9 A B
\dots\ Z a b \dots\ z ! " \#} {\$} \% \& ' ( ) * - / : ; \textless\ =
\textgreater\ ? @ [ {$\backslash$} ] \^{} \_ {`} {\{}
{\textbar} \} \textasciitilde} \cite{pmmm05a}. When all of
them are exhausted one reuses them prefixed by `+',
then again by `++', and so forth. An MMPH with $k$
vertices and $l$ edges we denote as a $k$-$l$ set.
In its graphical representation, vertices are depicted as dots and
edges as straight or curved lines connecting orthogonal vertices. 
In its ASCII string representation (used for computer processing)
each MMPH is encoded in a single line followed by assignments of
coordinatization to $k$ vertices. We handle MMP hypergraphs by
means of algorithms in the programs
SHORTD, MMPSTRIP, MMPSUBGRAPH, VECFIND, STATES01, and others
\cite{bdm-ndm-mp-1,pmmm05a-corr,pmm-2-10,bdm-ndm-mp-fresl-jmp-10,mfwap-s-11,mp-nm-pka-mw-11}.  

\subsection{\label{subsec:yuohks}KS vs.~operator
  contextuality}

Let us start with the {\em Kochen-Specker} theorem and KS sets.
Then we shall connect them with the vectors and operators of one
type of operator-based contextuality introduced by Yu and Oh. 

\begin{theorem}\label{th:ks} ({\em Kochen-Specker 
{\rm \cite{gleason,koch-speck,zimba-penrose}}})
In ${\mathscr{H}}^n$, $n\ge 3$, there are sets of $n$-tuples of mutually
orthogonal vectors to which it is impossible to assign 1s and 0s
in such a way that
\begin{enumerate}
\item No two orthogonal vectors are both
assigned the value 1;
\item In any group of $n$ mutually orthogonal vectors, not all of
the vectors are assigned the value 0.
\end{enumerate}
The sets of such vectors are called {\em KS sets\/} and the vectors
themselves are called {\em KS vectors\/}.
\end{theorem}

There is a one-to-one correspondence between KS $n$-tuples of vectors
and MMPH edges when they are all of their maximal size, as established
in \cite{pmmm05a,pavicic-pra-17,pm-entropy18,pwma-19}, and between KS
vectors and MMPH vertices with coordinatization within an MMPH with
maximal edges.

\begin{theorem}\label{th:mmph-nb}
  An $n$-dim {\rm MMPH} non-binary set with a coordinatization whose
  each edge contains $n$ vertices, is a {\rm KS} set.
\end{theorem}

\begin{proof} It follows straightforwardly from the KS theorem, its
  definition of a KS set and the aforementioned correspondences
  between its vectors and MMPH vertices.
\end{proof}

In 1988 Asher Peres presented a simple proof of the KS theorem
in a 3-dim Hilbert space using real vectors \cite{peres}.
He implicitly made use of 57 vectors/rays and 40 triplets of mutually
orthogonal vectors but seemed to have dropped 24 vectors that appear
in only one triplet and called his proof a ``33 vector [ray] proof.''
However, he admitted the role of the remaining vectors:
``It can be shown that if a single ray is deleted from the set of 33,
the contradiction disappears. It is so even if the deleted ray is not
explicitly listed in table 1.'' \cite[L176, bottom paragraph]{peres}.
From \cite[Table 1]{peres} we can reconstruct the 33 vectors
within their 40 triplets together with the ``non-explicit'' 24 vectors
and represent them in our MMPH notation, obtaining an MMPH non-binary
set with 57 vertices (vectors) and 40 edges (triplets), i.e., a 57-40 KS
set. We did so in two different ways with two resulting (but isomorphic)
hypergraphs in \cite[Fig.~4]{bdm-ndm-mp-fresl-jmp-10} and 
\cite[Fig.~19]{pavicic-pra-17}. Here we give a third MMPH
representation (isomorphic to the previous two) which contains the
so-called full scale Yu-Oh set
{\tt 123,345,567,789,9AB,BCD,DEF,FGH,HI1,1JK,KLA, 5LF,JPD,JM7,3OB,HN9.}
we elaborate on below. The representation is carried out via 
our programs SUBGRAPH and LOOP \cite{pm-entropy18}. 

Peres' 57-40 MMPH KS set reads:

\parindent=0pt

{\tt 123,345,567,789,9AB,BCD,DEF,FGH,HI1,1JK,KLA,
  JM7,3BO,H9N,JPD,FL5,QRS,STA,AUV,VWX,XYO,OZa,
  abc,cdC,CeQ,Sha,QgX,Vfc,bg9,qmU,Nnq,Bij,jku,
  klN,ur8,8st,iqt,Tpk,Tot,uvU.}

\parindent=10pt

Its graphical representation is given in Fig.~\ref{fig:yu-oh-peres}(a).

\begin{figure}[hbt]
\begin{center}
  \includegraphics[width=0.49\textwidth]{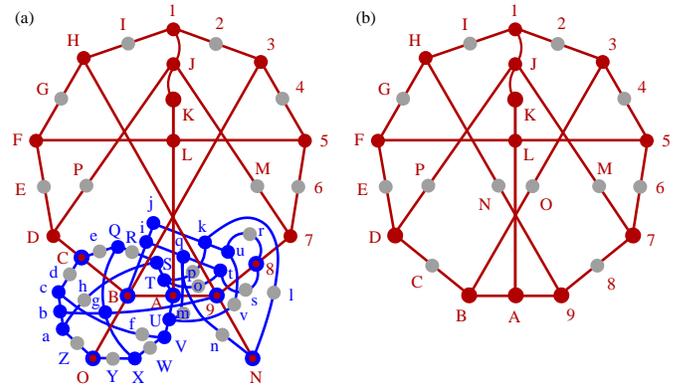}
\end{center}
\caption{(a) Peres' KS set 57-40 in the MMPH representation and
  containing the full scale Yu-Oh set (drawn in red); (b) The full
  scale Yu-Oh non-KS set 25-16;  Vertices (vectors) that share only
  one edge (triplet) are given as gray dots. See text.}
\label{fig:yu-oh-peres}
\end{figure}

Notice that gray dots 8,D,N,O in
Fig.\ref{fig:yu-oh-peres}(b) are not gray in
Fig.\ref{fig:yu-oh-peres}(a) and therefore the representation
of the original full scale 57-40  Peres KS set (with all gray dots
included) by means of the three original Yu-Oh non-KS sets (with gray
vertices dropped), as depicted in Fig.~1 of \cite{bengtsson-12},
apparently does not work. Also, as verified with our program SUBGRAPH,
Yu-Oh's set is not a subgraph  of Peres' 33-40 set (with all gray
dots dropped). On the other hand, Yu-Oh's set cannot be a subgraph of
Peres' 57-40 because it lacks gray dots. The full scale Yu-Oh's set
25-16 shown Fig.\ref{fig:yu-oh-peres}(b) is, of course, a subgraph
of the full-scale Peres' 57-40 set as shown in
Fig.\ref{fig:yu-oh-peres}(a) an confirmed by SUBGRAPH.

The arguments that all vertices are indispensable for an experimental
implementation of a KS set can be found in
\cite[In particular Table on p.~804]{larsson},
\cite[pp.~1583 top, 1588 bottom, and top 1589]{pmmm04c},
and \cite[p.~332, end of the 1st par.]{held-09}. In essence,
every $n$-tuple from the KS Theorem \ref{th:ks} should contain no
less than $n$ vectors. 

Below, the coordinatization of Peres' 57-40 set is obtained via
VECFIND \cite{pm-entropy18} from the vector components
$0,\pm 1,\sqrt{2}$ (the component $-\sqrt{2}$, used by Peres in
\cite{peres} is not needed):

\parindent=0pt
1=\{1,$\sqrt 2$,-1\},3=\{0,1,$\sqrt 2$\},5=\{-1,$\sqrt 2$,-1\},7=\{$\sqrt 2$,1,0\},

8=\{-1,$\sqrt 2$,0\},9=\{0,0,1\},A=\{0,1,0\},B=\{1,0,0\},
C=\{0,$\sqrt 2$,1\},D=\{0,-1,$\sqrt 2$\},F=\{1,$\sqrt 2$,1\},H=\{$\sqrt 2$,-1,0\},

J=\{-1,$\sqrt 2$,1\},K=\{1,0,1\},L=\{1,0,-1\},N=\{1,$\sqrt 2$,0\},

O=\{0,$\sqrt 2$,-1\},Q=\{-1,-1,$\sqrt 2$\},S=\{$\sqrt 2$,0,1\},T=\{-1,0,$\sqrt 2$\},

U=\{1,0,$\sqrt 2$\},V=\{$\sqrt 2$,0,-1\},X=\{1,1,$\sqrt 2$\},a=\{-1,1,$\sqrt 2$\},

b=\{1,1,0\},c=\{1,-1,$\sqrt 2$\},g=\{1,-1,0\},i=\{0,1,-1\},j=\{0,1,1\},

k=\{$\sqrt 2$,-1,1\},q=\{$\sqrt 2$,-1,-1\},t=\{$\sqrt 2$,1,1\},u=\{$\sqrt 2$,1,-1\}

\parindent=10pt

The aforementioned Peres' statement ``if a single ray is deleted
from the set of 33, the contradiction disappears'' amounts to a
coarse definition of a {\em vertex-critical} KS set:
``A KS [set] is termed critical iff it cannot be made smaller
by deleting the [vertices]'' \cite{ruuge12}. However, in KS
sets, there are edges whose removal does not
remove any vertex (but nevertheless cause a disappearance of
the KS property) and, on the other hand, no vertex can be
removed from a KS set without removing at least one edge as well,
in the sense that all edges/$n$-tuples should contain $n$ mutually
orthogonal vertices/vectors.

Therefore, we adopt a definition of an {\em edge-critical}
KS set \cite{pmm-2-10,bdm-ndm-mp-fresl-jmp-10,pavicic-pra-17}
(MMPH sets will require a redefinition of critical sets, as we
shall see later on):

\begin{definition}\label{def:crit} 
KS sets that do not properly contain any KS subset, meaning
that if any of its edges were removed, they would stop being KS 
sets, are called {\em critical\/} KS sets. 
\end{definition}

Hence, the set 
{\tt 13,35,57,79,9AB,BD,DF,FH,H1,1JK, KLA,5LF,JD,J7,3B,H9.}
Yu and Oh obtained in \cite{yu-oh-12} cannot be a KS set since
it is a subgraph of a critical KS set (Peres' set) and therefore
cannot provide a proof of the KS theorem contrary to the claim
in the title of \cite{yu-oh-12}, as we also
show in some detail in \cite[Sec.~XII]{pavicic-pra-17}. But,
in \cite{yu-oh-12}, Yu and Oh do define a new kind of contextuality
which we shall analyse and which we sumarize as follows.

Consider the MMPH of the Yu-Oh representation of the MMPH
Peres' subgraph, from Fig.~\ref{fig:yu-oh-peres}(b), shown in
Fig.~\ref{fig:yu-oh}. They removed all the vertices that share
only one edge and which are depicted as gray dots in
Fig.~\ref{fig:yu-oh}(a). Then they define operators by means of the
remaining vertices/vectors/rays/states which serve them to define
filters either for preparation or for detection of arbitrary
input or output states, respectively. The procedure goes as
follows.

\begin{figure*}[hbt]
\begin{center}
  \includegraphics[width=0.99\textwidth]{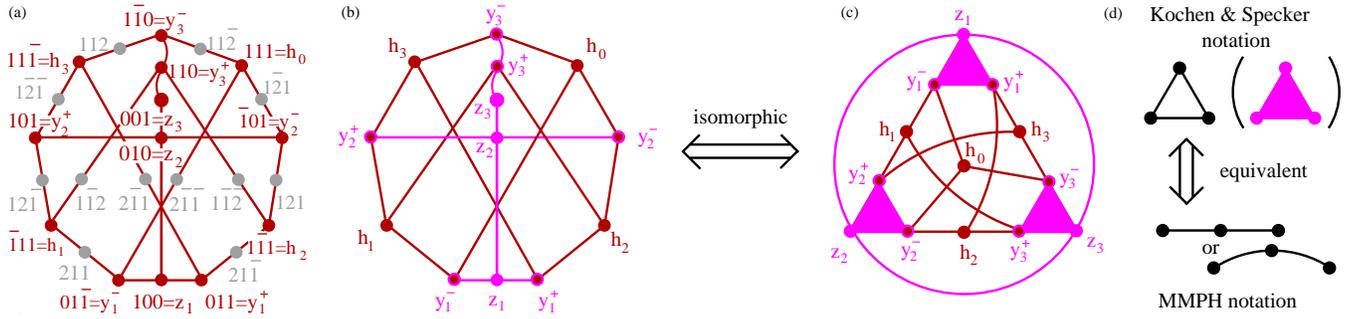}
\end{center}
\caption{(a) An MMPH subgraph of Peres' KS MMPH; 
  (b) Yu-Oh's reduction of (a); (c) Yu-Oh's Fig.~2 from
  \cite{yu-oh-12}; (d) Yu and Oh adopted a mixture of Kochen \&\
  Specker notation \cite{koch-speck};
  Cf.~\cite[Fig.~19]{pavicic-pra-17} (the triangles in (c)) and
  MMPH notation (the circle in (c)).} 
\label{fig:yu-oh}
\end{figure*}

Some of the vectors from Fig.~\ref{fig:yu-oh}(a) are represented as 

\begin{widetext}
\begin{align}
|y^-_1\rangle=
  \frac{1}{\sqrt{2}}
  \left(\begin{matrix}
0\\
1\\
-1\\
\end{matrix}  \right),\qquad
  |h_2\rangle=\frac{1}{\sqrt{3}}
  \left(
\begin{matrix}
1\\
-1\\
1\\
\end{matrix}  \right),\qquad
  |z_3\rangle=\left(
\begin{matrix}
0\\
0\\
1\\
\end{matrix}  \right),\qquad
  |N\rangle=\frac{1}{\sqrt{6}}
  \left(
\begin{matrix}
2\\
-1\\
1\\
\end{matrix}  \right).\qquad
\label{eq:vec}
\end{align}
\end{widetext}

Vectors serve Yu and Oh to define the following operators
\begin{align}
\hat{A}_i=I - 2|i\rangle\langle i| 
\label{eq:Ai}
\end{align}
where $i=1,\dots,13$ correspond to $y^-_1,y^-_2,\dots,z_3$ and
we add $i=14,\dots,25$ corresponding to gray dots in
Fig.~\ref{fig:yu-oh}(a). 
For instance, for $i=1,8,13,20$, corresponding to vectors from
Eq.~(\ref{eq:vec}), we have:

\begin{widetext}
\begin{align}
\hat{A}_1=
    \left(\begin{matrix}
1 & 0 & 0 \\
0 & 0 & 1 \\
0 & 1 & 0 \\
\end{matrix}  \right),\ \  
\hat{A}_8=\frac{1}{3}
  \left(
\begin{matrix}
1 & 2 & -2 \\
2 & 1 & 2 \\
-2 & 2 & 1 \\
\end{matrix}  \right),\ \ 
\hat{A}_{13}=\left(
\begin{matrix}
1 & 0 & 0 \\
0 & 1 & 0 \\
0 & 0 & -1 \\
\end{matrix}  \right),\ \ 
\hat{A}_{20}=\frac{1}{3}
  \left(
\begin{matrix}
-1 & 2 & 2 \\
2 & 2 & 1 \\
-2 & 1 & 2 \\
\end{matrix}  \right). \ 
\label{eq:A1813}
\end{align}
\end{widetext}

The operators can be combined in the following way:
\begin{align}
  \hat{L}_{13}=\sum^{13}_i\hat{A}_i-
  \frac{1}{4}\sum^{13}_i\sum^{13}_j\Gamma_{ij}\hat{A}_i\hat{A}_j=\frac{25}{3}I=8.\dot{3}I,
\label{eq:L}
\end{align}
where $\Gamma_{ij}=1$ whenever corresponding vectors $i,j$ are
orthogonal to each other and $\Gamma_{ij}=0$ when they are not;
also $\Gamma_{ii}=0$. The value 25/3 is curious since it is also
the sum of probabilities of detecting photons in the full scale
setup 25-16 shown in Fig.~\ref{fig:yu-oh-peres}(b). That may be
purely accidental. Also $\hat{L}_{25}$ is not diagonal.  
Yu and Oh consider neither vectors $|i\rangle$ nor operators
$\hat{A}_i$ for $i=14,\dots,25$ 

The fact that  each $\hat{A}_i$ has the spectrum $\{-1,1,1\}$ 
prompted Yu-Oh to calculate the upper bound of a corresponding
expression for 13 classical variables with predetermined values
-1 and 1:
\begin{align}
C_{13}=\sum^{13}_ia_i-\frac{1}{4}\sum^{13}_i\sum^{13}_j\Gamma_{ij}a_ia_j\le 8
\label{eq:C}
\end{align}

The inequality
\begin{align}
\langle\hat{L}\rangle>Max[C]
\label{eq:ine}
\end{align}
has been verified experimentally \cite{zu-wang-duan-12,li-zeng-19}
and also improved theoretically by changing the coefficients in
Eqs.~(\ref{eq:L}) and (\ref{eq:C})
\cite{cabello-bengtsson-12,kleinmann-cabello-12}.
However, no other set, apart from Yu-Oh's 13-16 itself, with such
properties has been found since.

We tested 50 sets and found that $\hat{L}$ of MMPHs without
left right symmetry mostly do not have diagonal matrices, although
some do, and that $\hat{L}$s of the majority of symmetric MMPHs
are also not diagonal; when they are, they are often not multiples
of $I$; for the ones whose $\hat{L}$s are multiples of $I$ we
found that they satisfy either $\langle\hat{L}\rangle<Max[C]$
or at most $\langle\hat{L}\rangle=Max[C]$, i.e., we have not
found instances of Eq.~(\ref{eq:ine}) being satisfied.
We give some examples below.

We should stress here that our definition of a {\em subgraph}
differs from a standard one. The standard definition assumes
that a subgraph is a hypergraph contained in a bigger
hypergraph as is. In contradistinction, we shall assume that
a subgraph might also be a hypergraph obtained from a bigger
hypergraph by taking out some edges and connecting the
remaining edges together, or simply by taking out some vertices.
The latter subgraph we denote as $\overline{\rm subgraph}$.
For instance {\tt 123,345,567.} is a standard subgraph of 
{\tt 123,345,567,781.}, while {\tt 123,345,561.} and
{\tt 13,345,567,781.} are its $\overline{\rm subgraphs}$.
Yu-Oh's 13-16 set is a $\overline{\rm subgraph}$ of Peres'
full scale 57-40 set. It is not a subgraph of either Peres'
57-40 or Peres' 33-40.

For a symmetric Kochen \&\ Specker's divided hexagon
\cite[Fig.~6(ii)]{pmmm05a} MMPH 8-7, a subgraph of the KS set
117-118 \cite{koch-speck}, shown in Fig.~\ref{fig:d}(a), we obtain
$\langle\hat{L}_8\rangle=Max[C_8]=9/2$.
The contextuality of the set has previously been considered in
\cite{clifton-93}.

\begin{figure*}[hbt]
\begin{center}
  \includegraphics[width=0.99\textwidth]{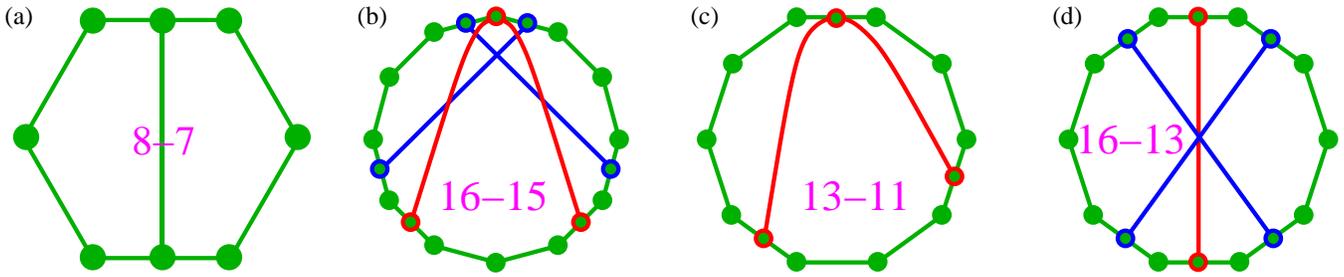}
\end{center}
\caption{(a) Hexagon MMPH from the KS set 192(117)-118
  \cite[Fig.~6(ii)]{pmmm05a} where it appears in 15 instances;
  (b) a symmetric $\overline{\rm subgraph}$ of Peres' MMPH
  with a non-diagonal $\hat{L}$; (c) an asymmetric
  $\overline{\rm subgraph}$ of Peres' MMPH with
  a diagonal $\hat{L}$ and $\langle\hat{L}\rangle<Max[C]$;
  (d) a constructed symmetric MMPH with a diagonal $\hat{L}$
  and $\langle\hat{L}\rangle<Max[C]$ but whose full scale
  version does not have a coordinatization.}
\label{fig:d}
\end{figure*}

From Peres' original KS set, using our programs STATES01, 
LOOP, and VECFIND we can generate arbitrary many subsets.
Most of them are asymmetric and their $\hat{L}$s are
non-diagonal. Also, many of highly symmetric ones, as,
e.g., 16-15 shown in Fig.~\ref{fig:d}(b) with $\hat{L}_{16}$
given in Eq.~(\ref{eq:L16}), are not diagonal. 

\begin{align}
\hat{L}_{16}=
 \frac{1}{6}   \left(\begin{matrix}
57 & 4 & 4 \\
4 & 54 & 3 \\
4 & 3 & 60 \\
\end{matrix}  \right)
\label{eq:L16}
\end{align}

An example of a non-symmetric 13-11 with a diagonal $\hat{L}$
is given in Fig.~\ref{fig:d}(c). It has $\langle\hat{L}_{13}\rangle=7.5$
and $Max[C_{13}]=7.75$, i.e., $\langle\hat{L}\rangle<Max[C]$.

We might try to construct a symmetric MMPH, e.g., the 16-13
one given in Fig.~\ref{fig:d}(d). For it we obtain
$\langle\hat{L}_{13}\rangle=9.5$
and $Max[C_{13}]=9.75$, i.e., again
$\langle\hat{L}\rangle<Max[C]$. However, the main problem
with such constructed MMPHs is that the probability of coming
across their filled (full scale) versions with coordinatizations
and therefore belonging to the 3-dim Hilbert space is minute, i.e.,
negligible even via automated construction and search on
a supercomputer. The full scale version (23-13) of the
aforementioned 16-13 apparently does not have a
coordinatization, either. 

We give more examples of $\langle\hat{L}\rangle$ vs.~$Max[C]$
calculations for other MMPHs in Sec.~\ref{subsec:crit}.

\subsection{\label{subsec:masters}MMPH masters}

There are several facts we would like to stress as starting
points of our elaboration on the MMPH non-binary sets.

\begin{enumerate}
\item[{(\em i)}] Peres wrote: “It can be shown that if a single ray is
  deleted from the set of 33, the contradiction disappears.
  It is so even if the deleted ray is not explicitly listed
  in table 1.” \cite[L176, bottom paragraph]{peres}
\item[Ad~{(\em i)}]The first sentence is wrong because
  MMPH 33-40 set {\tt 123,345,47,79,92A,AC,C4,AF,5F,HJ,
    HL,H7M,NCO,OPQ,QRL,RT,TJ,JPV,VX,XR,Va,La,
    ce,cT1,cg,FXM,Mhi,ijg,jl,le,ehn,np,pj,nN,
    gN,t9,tlO,t5,ap1,1MO.} is not critical as verified by
  STATES01. It is also not a KS set but only an MMPH non-binary set.
  The second sentence is conditionally correct because the full scale
  MMPH 57-40,
  {\tt 123,345,467,789,92A,ABC,CD4,AEF,5GF,
    HIJ,HKL,H7M,NCO.OPQ,QRL,RST,TUJ,JPV,VWX,
    XYR,VZa,Lba,cde,cT1,cfg,FXM,Mhi,ijg,jkl,
    lme,ehn,nop,pqj,nrN,gsN,tu9,tlO,tv5,ap1, 1MO.}
  is a critical KS set but only if assume that with the deleted ray
  we also delete the edge/triplet it belonged to. (This instance of
  Peres' 57-40 KS set is isomorphic to the one given above; the
  sequence of characters is different due to a reshuffling by
  automated tools we used to obtain 33-40 as a subgraph of 57-40. 
\item[{\em (ii)}] Yu and Oh write: ``The KS value assignments to
  the 13-ray set [13-16] are possible; i.e., no logical
  contradiction can be extracted by considering conditions 1 and 2
  [of Theorem \ref{th:ks}].'' \cite[p.~3, left column, top]{yu-oh-12}
\item[Ad {\em (ii)}]The claim is provisionally correct, but not 
  because ``no logical contradiction can be extracted by considering
  conditions 1 and 2''---it can be extracted: in 13-16 it is
  impossible to assign 1s and 0s in such a way that conditions 1 and 2
  are satisfied---and not because ``value assignments to the 13-ray
  set are possible''---they are not possible: one cannot assign 1s
  and 0s to its rays in such a way that conditions 1 and 2 are
  satisfied---but because the 13-16 set is not a set of triplets and
  therefore does not satisfy the first part of the KS theorem. 
\end{enumerate}

The ``incomplete triplets'' issue reappears in many papers and
books. For instance in Karl Svozil's book \cite{svozil-book-ql} in
Sec.~7.4 there is an excellent symmetric figure of Peres' 33-40 set
[Fig.7.12], we, actually, made use of to write down MMPH 57-40 set,
but we had to add 24 vertices that were not there; 33 vectors and
their corresponding logical proposition were explicitly given, but
the remaining 24 vectors were not mentioned. In the original
Kochen-Specker paper \cite{koch-speck} the triplets (edges with
3 vertices) were depicted as triangles and doublets (triplets from
which one vertex was dropped) as straight lines---all together
117 vertices of 192 ones contained in 118 triplets. Their triangles
are shown in \cite[Fig.~6(ii)]{pmmm05a}. The same triangles are
used in the Yu-Oh's set and are shown in Fig.~\ref{fig:yu-oh}(d).
This triangle notation is a source of some confusion in the
literature and research, though. For instance, in 
\cite{cabell-klein-budr-prl-14} on p.~4, Fig.~1 (b), where one line
from one of the triangles from Yu-Oh's set is deleted, we read:
``(b) $G_{\rm YO}$ minus one edge.'' However, the lines in the
triangle are not edges. The whole triangle is an edge (triplet) as
shown in  Fig.~\ref{fig:yu-oh}(d). The lines within a triangle are
orthogonalities and a removal of one of them means splitting the
triplet into two doublets, i.e., increasing the number of edges in
the set. So, the set in Fig.~1(a) of \cite{cabell-klein-budr-prl-14}
has 16 edges, while the set in Fig.~1(b) has 17 edges. In any case
the set (b) is not a subgraph of (a) nor is (a) a subgraph of (b). 
Of course, a removal of one of the orthogonalities must also be
accompanied by a switch to a new coordinatization of the whole set. 

In the {\em The Kochen-Specker Theorem} article in the 
{\em Stanford Encyclopedia of Philosophy} only 117 vertices were
considered. ``[W]hat KS have shown is that a set of 117 yes-no
observables cannot consistently be assigned 0-1 values''
\cite{ks-stanford-enc-18}. Jeffrey Bub writes:
``This yields a total of 49 rays and 36 orthogonal triples. Now the
only rays that occur in only one orthogonal triple are the 16 rays
with a 5 as component. Removing these 16 rays from the 49 rays
yields the following set of 33 rays that cannot be colored.''
\cite{bub}. However, 49 rays also cannot be colored and the
49-36 is critical, while 33-36 is not. 

These facts offer the following approach, though. The aforementioned
conditions 1 and 2 are also contained in the Def.~\ref{df:mmphs} of
an MMPH non-binary set and Peres' 33-40, Yu-Oh's 13-16, Bub's 33-36,
Conway-Kochen's 31-37, and Kochen-Specker' 117-118 sets all violate
the conditions 1 and 2, thus confirming that these sets are MMPH
non-binary sets. Moreover, they actually enable us to get many
smaller MMPH non-binary sets from them because none of these sets
is critical. And they are all equipped with at least the
coordinatization they inherit from their full scaled versions 57-40,
25-16, 49-36, 51-37, and 192-118, respectively, but often with even
simpler ones.

The MMPH strings of the last three sets are:

\parindent=0pt
Bub's 33-36 (derived from the full scale 49-36
\cite[Fig.~19]{pavicic-pra-17}):~{\tt 12,134,156,67,48,9AB,CDE,6B,4E,2FG,2HI,
  EG,GB,8I,I7,AJ,AK,C7L,MN9,HON,N3P,PL,MFQ,QL,
  M5R,RD,DO,STC,JHT,T5U,S3K,SFV,VW,98W,WU,X9C.}

Conway-Kochen's 31-37 (derived from the full scale 51-37
\cite[Fig.~19]{pavicic-pra-17}):
{\tt 123,245,26,57,89A,BCD,5D,3EF,3G,DF,
  FA,9H,87I,9J,CK,CL,LM,HN,M1N,KO,1OP,Q6R,QGH,BQS,
  PR,PJ,S4J,SET,NT,TI,RI,UV8,VGK,U6L,4V,UE,18B.}

and the Kochen-Specker's 117-118 (derived from the original
full scale 192-118 \cite[Fig.~19]{pavicic-pra-17}):

\begin{widetext}  
  {\tt 12,234,45,56,678,81,9A,ABC,CD,DE,EFG,G9,HI,IJK,KL,LM,MNO,OH,PQ,QRS,ST,TU,UVW,WP,1X,XYZ,Za,ab,bcd,
    d1,ef,fgh,hi,ij,jkl,le,mn,nop,pq,qr,rst,tm,uv,vwx,xy,yz,z!","u,\#\$,\$\%\&,\&','(,()*,*\#,e-,-/:,:;,;<,
    <=>,>e,?@,@[{$\backslash$},{$\backslash$}],]\^{},\^{}\_`,`?,\{|,|\}\textasciitilde,\textasciitilde +1,+1+2,+2+3+4,+4\{,+5+6,+6+7+8,+8+9,+9+A,+A+B+C,+C+5,+D+E,
    +E+F+G,+G+H,+H+I,+I+J+K,+K+D,?+L,+L+M+N,+N+O,+O+P,+P+Q+R,+R?,37,BF,JN,RV,Yc,gk,os,w!,\%),/=,[\_,
    \}+3,+7+B,+F+J,+M+Q,95e,HDe,PLe,aTe,mi?,uq?,y'?,;\#?,\{]1,+5+11,+D+91,+O+H1,1e?.}
\end{widetext}

\parindent=10pt
All of them have coordinatizations and none of them is critical. They
will be our MMPH non-binary {\em master sets} we shall get smaller
MMPH non-binary critical sets from, in Sec.~\ref{subsec:crit}.
Here we want to stress that we have chosen the above sets to be
our masters for historical reasons. But any set we obtain by stripping
the original four KS sets from some other number of vertices being
contained in only one edge can serve us as a master set. For instance,
by stripping not 24 but 17 such vertices from Peres' 57-40 KS set, we
obtain the following set  which we can also use as our master
set---Peres' 40-40  (derived from the full scale 57-40
\cite[Fig.~19]{pavicic-pra-17}): {\tt 123,345,467,78,829,9A,A4,9B,5B,CD,CE,
  C7F,GAH,HIJ,JKE,KLM,MND,DIO,OPQ,QRK,OST,ET,UVW,
  UM1,UX,BQF,FYZ,ZaX,ab,bW,WYc,cd,da,cG,XG,e8, ebH,e5,Td1,1FH.}

\parindent=10pt
We present two smaller critical MMPH non-binary sets 35-27 and 38-30,
obtained from this 40-40 set, in Appendix \ref{app:1c} because they
are bigger than Peres' 33-40 and they are critical, while Peres' 33-40
is not. Also, criticals with 33 or less vertices we obtained from
Peres' 33-40 and from Peres' 40-40 coincide. The difference is only in
criticals with 34 to 38 vertices which we, of course, cannot obtain
from Peres' 33-40 set.

\subsection{\label{subsec:crit}Classes of MMPH non-binary
  sets, their implementation, and their inequalities}

From the MMPH non-binary master sets given in
Sec.~\ref{subsec:masters} we obtain smaller MMPH non-binary
critical sets via STATES01. There is a principal difference in the
feature of criticality between these sets and the full scale KS sets,
though. 

If we removed any of the edges of a full scale KS critical set,
the remaining set would not be a KS set any more (see
Def.~\ref{def:crit}). If we then continued to strip further
edges from the remaining set, we would never arrive at a KS
set again. This is not so with an MMPH non-binary critical set.
When we remove any of its edges it does stop being an MMPH
non-binary set, but if we removed further edges from the
obtained set, it would often turn into a smaller MMPH non-binary
critical set. Therefore we introduce:

\begin{definition}\label{def:crit-nb} 
  An MMPH non-binary set is called an MMPH non-binary
  {\em critical\/} sets if a removal of any of its edges would
  turn the remaining set into an MMPH binary set. MMPH non-binary
  critical sets might properly contain smaller MMPH non-binary
  critical sets whose number of edges is smaller than the original
  critical set for at least 2 edges.
\end{definition}

Bub and Conway-Kochen's master sets share the coordinatization
while Peres and Kochen-Specker's ones have different ones
mutually and with respect to the former two sets. Therefore, also
the classes of smaller MMPH non-binary critical sets we obtain
from them will be structurally different.

From these master sets we generated classes of
smaller MMPH non-binary critical sets by means of our programs
\cite{pmmm05a,pm-entropy18}, although the algorithms and
programs should be redesigned and rewritten for an automated
generation. MMPH sets generated from a master set we call
a class of MMPH sets. So, we shall talk about Bub,
Conway-Kochen, Peres, and Kochen-Specker's classes.
Distributions of their criticals are shown in
Fig.~\ref{fig:dis-strip}. The criticals are mostly the standard
subgraphs of their masters obtained via our automated algorithms
and programs, except for a limited number of smaller
$\overline{\rm subgraphs}$  we obtained via new algorithms which
are still under development. Most $\overline{\rm subgraphs}$
have a parity proof unlike most of the standard subgraphs of which
only a very few have a parity proof. 

Notice that the biggest critical sets in
Fig.~\ref{fig:dis-strip}(a,c) have the same number of vertices as
their master sets, but 9,12 edges less, respectively. 

\begin{figure*}[hbt]
\begin{center}
  \includegraphics[width=0.9\textwidth]{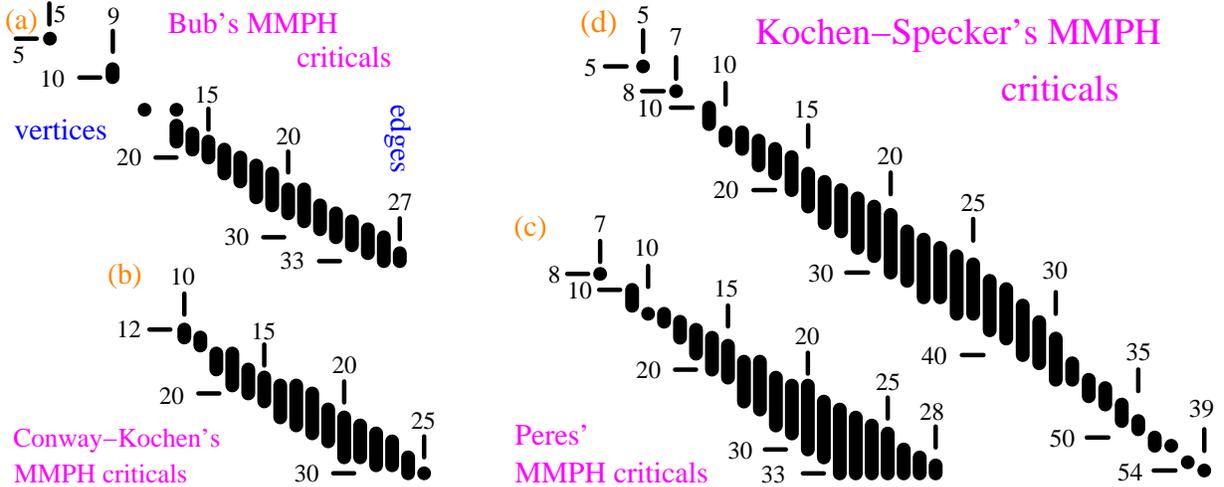}
\end{center}
\caption{(a) Distribution of MMPH non-binary critical sets
  generated from Bub's MMPH non-binary master set;
  (b) Conway-Kochen's criticals; (c) Peres' criticals;
  (d) Kochen-Specker's criticals.}
\label{fig:dis-strip}
\end{figure*}

A possible experimental implementation of MMPH non-binary
sets might be made in analogy to the experimental implementation
of KS sets carried out in \cite{d-ambrosio-cabello-13}. The
difference is that the latter sets contain only triplets, while
the former ones contain triplets and doublets, similarly to the
Yu-Oh's 13-16 set, or even only doublets as in the 5-5 set. To
carry out measurements on KS sets means that we have to verify
that the probability of detecting a particle or a photon at each
out-port of a gate representing an edge (triplet) is 1/3. Yu-Oh's
implementation rely on gates defined via Eqs.~(\ref{eq:Ai})
and (\ref{eq:L}) by means of 13 vertices/vectors/rays/states
and the gates representing 12 dropped vertices are not considered.
Measurements on MMPH non-binary sets might be carried out as for
KS sets with triplets (with the 1/3 probability of detection
at each out-port) and via {\em calibrated} detections of a
particle or a photon at out-ports of a gate representing a doublet
with the 1/2 probability of detecting a particle at each of the
two considered ports. When a vertex share a mixture of triplet and
doublet edges the probability of detection is $1/p$, where
$1/3\le p\le 1/2$. The data obtained at the out-ports corresponding
to the dropped third vertices are discarded or we simply do not
measure them at all as in Yu-Oh's experiments 
\cite{zu-wang-duan-12,li-zeng-19,cabello-bengtsson-12}.
To assure an equal distribution of outcomes at each port,
the inputs to doublet gates should be scaled up with respect to the
full triplet ones by 3/2 and this is why we call them
{\em calibrated}. 

The inequalities to be experimentally verified for the MMPH
non-binary sets can be defined as for the other two kinds of sets.
For instance, for Yu-Oh's 13-16 set we verify their inequality
given by Eq.~(\ref{eq:ine}): $8.3>8$. Let us consider the set
as shown in Fig.~\ref{fig:yu-oh-peres}(b) (excluding the gray dots).
This set contains 4 triplets and 12 doublets. Vertices {\tt A,K,L}
share only triplets, so the probability of having a click
along them is 1/3. Vertices {\tt 3,7,D,H} share only doublets
and the probability of getting clicks along them is 1/2.
Vertices {\tt 1,5,9,B,F,J} share a triplet and two doublets, each,
what yields the probability $(1/2+1/2+1/3)/3=4/9$. Altogether,
the probabilities for 13 vertices sum up to 
$3\times 1/3+4\times 1/2+6\times 4/9=17/3$. Let us call this sum
a {\em quantum hypergraph index} of an MMPH set and denote it as
$HI_q$. On the other hand, the set 13-16 allows at most four 1s.
This is a classical upper bound for getting classical detection
clicks. Let us call this classical upper bound, i.r., the maximal
number of 1s we can assign to vertices of a MMPH non-binary set
so as to satisfy the two conditions from Def.~\ref{df:mmphs}, a
{\em classical hypergraph index} $HI_c$. Hence, we obtain the
inequality $HI_q[\text{13-16}]=17/3=5.\dot{6}>HI_c[\text{13-16}]=4$.
Notice, that even {\em uncalibrated} probabilities give us
$HI_{q-unc}[\text{13-16}]=13/3=4.\dot{3}>HI_c[\text{13-16}]=4$. We
obtain uncalibrated probabilities by measuring all vertices in all
edges in Fig.~\ref{fig:yu-oh-peres}(b), meaning with gray dots
included. With each vertex in every edge we have a probability of
getting a click, i.e., of assigning 1 to it, being equal to 1/3.
If we now selected the 13 red-dot vertices, we would get
$13/3=4.\dot{3}$ which is also greater than $HI_c[\text{13-16}]=4$.
Notice also that the maximal number of 1s we can assign to vertices
in the full scale 25-16 set is 11 and that gives us the inequality
$HI_q[\text{25-16}]=25/3=8.\dot{3}<HI_c[\text{25-16}]=11$ which is
yet another proof that 25-16 is not a KS set. 

It is interesting that three of four considered masters also
satisfy the uncalibrated inequality $HI_{q-unc}>HI_c$. Bub's 33-36:
$HI_{q-unc}[\text{33-36}]=11>HI_c[\text{33-36}]=10$, Conway-Kochen's
31-37: $HI_{q-unc}[\text{31-37}]=10.\dot{3}>HI_c[\text{31-37}]=8$,
and Peres' 33-40 $HI_{q-unc}[\text{33-40}]=11>HI_c[\text{33-40}]=6$.

Let us now present several smaller MMPH  criticals from each
class, consider their properties, and calculate 
Yu-Oh-like expressions and values for some of them.

The smallest Bub's critical $\overline{\rm subgraph}$ with
coordinatization we found is the pentagon 5-5 {\tt 12,23,34,45,51}
(with the gray dots excluded) shown in Fig.~\ref{fig:bub-c}(a).
The full scale hypergraph 10-5
{\tt 162,273,384,495,5A1} is also shown Fig.~\ref{fig:bub-c}(a)
(with the gray dots included).

\begin{figure*}[hbt]
\begin{center}
  \includegraphics[width=0.99\textwidth]{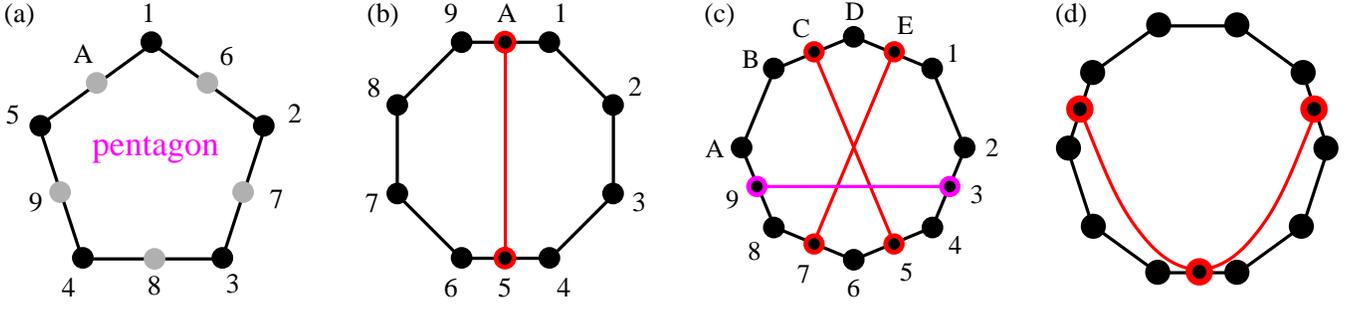}
\end{center}
\caption{Criticals generated from Bub's master:
  (a) $\overline{\rm subgraph}$ pentagon 5-5;
  (b) $\overline{\rm subgraph}$ 10-9;
  (c) standard subgraph 14-11;
  Critical generated from Peres' master:
  (d) 13-11.}
\label{fig:bub-c}
\end{figure*}

As we proved in \cite{pmmm05a}, the smallest loop edges can form
in a 3-dim space with vertices endowed with a real coordinatization
is a pentagon. We could not find (with Mathematica) a complex
coordinatization of any smaller MMPH, either. We conjecture that
the filled pentagon MMPH 10-5 is the smallest MMPH with a 
coordinatization in the 3-dim Hilbert space. Its coordinatization
is, e.g.,
1=\{0,0,1\},2=\{0,1,0\},
3=\{1,0,1\},4=\{1,1,-1\},5=\{1,-1,0\},6=\{1,0,0\},7=\{1,0,-1\},
8=\{-1,2,1\},9=\{1,1,2\},A=\{1,1,0\}.
It, of course, includes the coordinatization of
5-5. As we can easily check, the maximal number of 1s assignable to
vertices of 5-5, satisfying the two aforementioned condition, is 2.
Thus we have the following contextual inequality
$HI_q[\text{5-5}]=5\times 1/2=2.5>HI_c[\text{5-5}]=2$. Yu-Oh's
approach does not offer us such a contextual distinguisher since
for $\hat{L}$ and $C$ of Eqs.~(\ref{eq:L}), (\ref{eq:C}), and
(\ref{eq:ine}) we get $\hat{L}_{10}=2.5I$ and $C_{10}\le 2.5$.
Hence, $\langle\hat{L}_{10}\rangle=Max[C_{10}]$. MMPH non-binary
$\overline{\rm subgraph}$ 5-5 can actually be generated in all four
MMPH classes, but we have not shown them for Conway-Kochen and
Peres' classes in Fig.~\ref{fig:dis-strip}. The pentagon 5-5 has a
parity proof.

Subsequent small Bub's critical $\overline{\rm subgraphs}$ we
obtained, are 9-9 and 10-9. The latter is shown in
Fig.~\ref{fig:bub-c}(b). Its MMPH string can be easily read off
from the figure: {\tt 12,23,34,456,67,78,89,9A1,A5.}
Its possible coordinatization is:
1=\{0,0,1\},2=\{1,1,0\},
3=\{1,-1,1\},4=\{0,1,1\},5=\{2,-1,1\},6=\{1,1,-1\},7=\{1,0,1\},
8=\{1,2,-1\},9=\{2,-1,0\},A=\{1,2,0\}.
Vector component `2' is here because the set of {\tt 1-A} vertex
coordinates  is a subset of the {\tt 1-H} set of coordinates of the
filled set 17-9. As for the contextuality verification, we have
$HI_q[\text{10-9}]=6\times(1/2+1/3)/2+4\times 1/2=9/2=4.5\dot{3}>HI_c[\text{10-9}]=4$.
On the other hand, we have $\hat{L}_{10}=5.5I$ and $C_{10}\le 5.5$.
Hence, $\langle\hat{L}_{10}\rangle=Max[C_{10}]$. The set has a
parity proof. 

The first standard subgraph in the Bub's class we found is
14-11 shown in Fig.~\ref{fig:bub-c}(c). Its coordinatization~is
1=\{2,0,1\},2=\{-1,-1,2\},3=\{1,-1,0\},4=\{1,1,1\},5=$\!$\{2,-1,-1\},

\parindent=0pt
6=\{0,1,-1\},7=\{2,1,1\},8=\{-1,1,1\},9=\{1,1,0\},A=\{1,-1,2\},

B=\{2,0,-1\},C=\{1,0,2\},D=\{0,1,0\},E=\{-1,0,2\}. 

$HI_q[\text{14-11}]=4\times 1/3+10\times(1/2+1/3)/2=11/2=5.5\dot{3}>HI_c[\text{14-11}]=5$. The Yu-Oh approach gives: $\hat{L}_{10}=8.5I$ and $C_{10}\le 8.75$.
Hence, $\langle\hat{L}_{14}\rangle<Max[C_{14}]$. The set is
one of the few standard subgraphs that have a parity proof.  
The only other Bub's criticals with a parity proof we found
are 14-13, 18-15, 24-19, and 28-23.

\parindent=10pt
Another critical with $\hat{L}=cI$ ($c$ is a constant) we found
is 14-13: {\tt 12,23,34,45,56,67,789,9A,AB,BC,CD, DE1,E8.}
$\langle\hat{L}_{14}\rangle=7.5<Max[C_{14}]=7.75$.

Yu-Oh's 13-16 is from the Peres' class but the only other
critical with $\hat{L}=cI$ we found in Peres' class is the
$\overline{\rm subgraph}$ 13-11 shown in
Fig.~\ref{fig:dis-strip}(d):
{\tt 12,234,56,678,89,9A,ABC,CD,D1,35,4B7.} 
$\langle\hat{L}_{13}\rangle=7.5<Max[C_{13}]=7.75$,
The coordinatization is 1=\{1,1,$\sqrt{2}$\},2=\{0,$\sqrt{2}$,-1\},3=\{0,1,$\sqrt{2}$\}, 4=\{1,0,0\},

\parindent=0pt
5=\{1,$\sqrt{2}$,-1\},6=\{$\sqrt{2}$,-1,0\},7=\{0,0,1\},8=\{1,$\sqrt{2}$,0\},

9=\{$\sqrt{2}$,-1,1\},A=\{1,0,-$\sqrt{2}$\},B=\{0,1,0\},C=\{$\sqrt{2}$,0,1\},

D=\{1,1,-$\sqrt{2}$\}. The components $\pm\sqrt{2}$ come from the
coordinatization of the filled set 20-11 which requires 
the components $\pm\sqrt{2},3$, i.e., more than Peres'
master set itself. This is because 13-11 is a
$\overline{\rm subgraph}$ and not a standard subgraph of
the master set. $HI_q[\text{13-11}]=5.5\dot{3}>HI_c[\text{13-11}]=5$.
The critical 13-11 has a parity proof.
We found no standard subgraph of Peres' master with a
parity proof, though.

\parindent=10pt

In Fig.~\ref{fig:dis-strip}(b), only critical standard subgraphs
obtained via automated generation are shown. Hence, they are all
subgraphs of Conway-Kochen's master but we shall explain how one
can generate $\overline{\rm subgraphs}$ from them.

Let us consider Conway-Kochen's critical 13-10 shown in
Fig.~\ref{fig:c-k-c}(a): {\tt 12,234,45,56,678,89,9A1,ABC,
  3B7,CD5.}
Its coordinatization is: 1=\{1,1,0\},2=\{-1,1,1\},

\parindent=0pt
3=\{1,0,1\},4=\{1,2,-1\},5=\{0,1,2\},6=\{1,-2,1\},7=\{1,0,-1\},

8=\{1,1,1\},9=\{1,-1,0\},A=\{0,0,1\},B=\{0,1,0\},C=\{1,0,0\},

D=\{0,2,-1\}., after taking into account the filled
17-10 set. Similarly to Yu-Oh's set, the 13-10 set exhibits both
contextual indices: 
$HI_q[\text{13-10}]=4.9\dot{4}>HI_c[\text{13-10}]=4$
and 
$HI_{q-unc}[\text{13-10}]=13/3=4.\dot{3}>HI_c[\text{13-10}]=4$.
If we take out the vertex D (the gray dot in
Fig.~\ref{fig:c-k-c}(a)) the resulting $\overline{\rm subgraph}$
12-10 is critical too, which also shows that vertex-criticality
is not consistent. Unlike Yu-Oh's set, neither 13-10 nor 12-10
have $\hat{L}=cI$ satisfied. $\hat{L}_{13}$ is not diagonal and
$\hat{L}_{12}$ is diagonal but it is not a multiple of the unit
matrix. The set 12-10 does not exhibit both contextual
distinguishers:
$HI_q[\text{12-10}]=4.75\dot{4}>HI_c[\text{12-10}]=4$ but 
$HI_{q-unc}[\text{12-10}]=12/3=4=HI_c[\text{12-10}]=4$.
It is, of course, due to the lower number of vertices, since the
geometrical structure of the MMPHs, yielding the classical index 4,
remains the same.

\parindent=10pt

\begin{figure*}[hbt]
\begin{center}
  \includegraphics[width=0.9\textwidth]{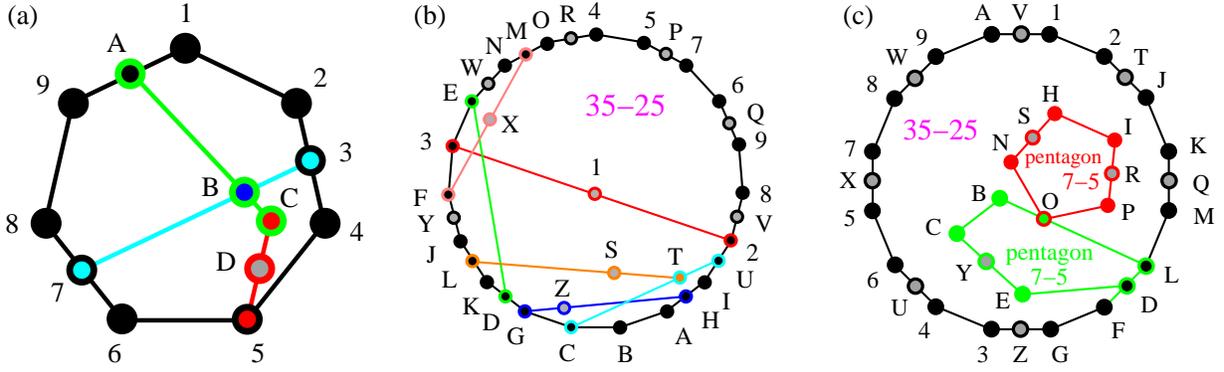}
\end{center}
\caption{(a) Conway-Kochen's MMPH non-binary critical
  set 13-10;
  (b) Kochen-Specker's 35-25a critical with uncalibrated
  contextuality; the outer loop is a 19-gon;
  (c) Kochen-Specker's 35-25b critical without uncalibrated
  contextuality; the outer loop is a 16-gon; See text.}
\label{fig:c-k-c}
\end{figure*}

We find similar features within Kochen-Specker's MMPH class.
Let us take two MMPH criticals from the middle of the distribution
shown in Fig.~\ref{fig:dis-strip}(d). 32-25a:
{\tt 45,5P7,76,6Q9,98,8V2,2UI,IHA,AB,BC,CG,GDK,KLJ,
  JYF,F3,3E,EWN,NMO,OR4,123,DE,STL,UTC,XMF,ZHG.}
and 35-25b: {\tt 12, 2TJ,JK,KQM,ML,LDF,FG,GZ3,34,4U6,
  65,5X7,78,8W9,9A,AV1,BC,DE,HI,NO,PO,RPI,SNH,
  YEC,OLB.} Their coordinatizations are too long to be given here.
Neither of them nor any other standard subgraph
in the Kochen-Specker's class we obtained in
Fig.~\ref{fig:dis-strip}(d) has a parity proof. 

Their different geometrical structure yield different classical
hypergraph indices: $HI_c[\text{35-25a}]=11$
and $HI_c[\text{35-25b}]=12$. However, the number of vertices
and therefore the quantum uncalibrated hypergraph indices of
both MMPHs are the same: $HI_{q-unc}[\text{35-25}]=35/3=11.\dot{6}$.
That means that 35-25a exhibits contextuality even for
uncalibrated measurement outcomes, while 35-25a does not.
Their calibrated indices are:
$HI_q[\text{35-25a}]=12.\dot{4}>HI_c[\text{35-25a}]=11$
and $HI_q[\text{35-25b}]=13.75>HI_c[\text{35-25b}]=12$.
Pentagons in 35-25b in 
Fig.~\ref{fig:c-k-c}(c) are subgraphs of Kochen-Specker's
master unlike the pentagon 5-5 (without gray dots) in
Fig.~\ref{fig:bub-c}, which is a $\overline{\rm subgraph}$.
If we removed all gray dots, the resulting set 25-25 will not be
critical any more, but if we leave {\tt S} and {\tt R} in the
red pentagon, the resulting 27-25 set will be critical. This
cannot be achieved with the green pentagon: leaving {\tt Y} as the
only gray dot in the 26-25 set will not make it critical.
$\hat{L}$ of the double pentagon is not diagonal. 

In Appendix \ref{app:1} we give chosen MMPH non-binary
critical sets which are standard subgraphs of the four  
MMPH master sets.

\section{\label{sec:disc}Discussion}

In the last half a century a vast number of constructive
proofs of quantum contextuality were obtained in even dimensional
Hilbert spaces, but only a very few in odd dimensional ones.
In particular, in the 3-dim space: Bub, Conway-Kochen, Peres,
and Kochen-Specker's KS sets, Yu-Oh contextual set, and
Klyachko-Can-Binicio{\u g}lu-Shumovsky's pentagram/pentagon
state-dependent set. All together 6 sets.

In this paper we present $n$-dim hypergraph contextuality
which consists in generating sets which preclude binary
assignments of values 0 and 1 to vertices of a hypergraph,
such that 1 is assigned to only one of the vertices in each
edge of the hypergraph, where an edge can contain less than $n$
mutually orthogonal vertices. Such a set which we call an $n$-dim
MMPH non-binary set, is defined by Def.~\ref{df:mmphs}.
We stay with $n=3$, i.e., we deal with qutrits
only, although the method can be extrapolated to any dimension.
The method serves us to distinguish classical models with
predetermined binary values, that can be assigned to measurement
outcomes of classical observables underlying classical
computation, from quantum models that do not allow for such values
and that underlie quantum computation.

Let us make use of a graphical representation of an $n$-dim MMPH to
describe the method. Vertices within an MMP hypergraph are drawn as
dots and edges containing mutually orthogonal vertices are drawn with
the help of straight or curved lines connecting these ``orthogonal
dots'' as shown in Figs.~\ref{fig:yu-oh-peres}, \ref{fig:yu-oh},
\ref{fig:d}, \ref{fig:bub-c}, and \ref{fig:c-k-c}. There can be a
different number of vertices/dots in edges.
Our program then verifies whether a chosen MMPH $k$-$l$ violates
or obeys the 0,1 assignment rules from Def.~\ref{df:mmphs}.
Edges in MMPH $k$-$l$ might contain 3 or 2 vertices.
We then consider a filled MMPH $k'$-$l$ in which we add a vertex
to each edge which contains only 2 vertices and try to find a
coordinatization for it. If successful, we make a one-to-one
correspondence between vertices and vectors in the $n$-dim Hilbert
space, i.e., for the MMPH $k'$-$l$ set. The MMPH $k$-$l$ set
inherits the coordinatization from from the MMPH $k'$-$l$ set. 
If we implemented the MMPH $k'$-$l$, each edge would be a gate with
$n$ outcomes and the probability of detecting an outcome would be
$1/n$.

Now, our approach consists in discarding the outcomes corresponding
to chosen vertices which share (are contained in) only one edge
from chosen edges and considering outcomes only of the remaining
vertices. In the 3-dim Hilbert space, that means that
some of the edges/gates should be taken as doublets and the others
as triplets. Our programs can handle such MMPHs because they are
written for edges of mixed sizes. Measurements on MMPH non-binary
sets might then be carried for triplets in a standard manner, i.e.,
with the probability of 1/3 of obtaining a click (value 1) at each of
the three ports, at a gate corresponding to an edge/triple, and via
a {\em calibrated} detection at out-ports of a gate representing a
doublet with the probability of 1/2. For vertices that share triplet
and doublet edges, the probability would be equal to $1/p$, where
$1/3\le p\le 1/2$. Calibration consists in sending three input
particles to a doublet gate for each two sent to a triplet gate,
i.e., the ratio of doublet to triplet inputs should be 3/2. 

To obtain a measure of quantum contextuality of an MMPH non-binary
set we define hypergraph indices. A classical hypergraph index
$HI_c$ is the maximal number of 1s we can assign to vertices within
edges of an MMPH so as to obey the 0,1 assignment rules from
Def.~\ref{df:mmphs}. A (calibrated) quantum hypergraph index $HI_q$
is the sum of calibrated probabilities for all $k$ vertices of
the aforementioned $k$-$l$ MMPH. An uncalibrated quantum hypergraph
index $HI_{q-unc}$ is the sum of 1/3-probabilities for all $k'$
vertices of the aforementioned $k'$-$l$ MMPH. A basic measure of
quantum contextuality of an MMPH non-binary set is the inequality
$HI_c<HI_q$. If it were satisfied, the MMPH would be contextual.
If not, it wouldn't. A stronger measure of quantum contextuality of
an MMPH non-binary set is the inequality $HI_c<HI_{q-unc}$. Some of
the considered MMPHs do satisfy both inequalities. For instance,
Yu-Oh's set 13-16, MMPH 13-10 shown in Fig.~\ref{fig:c-k-c}(a),
MMPH 35-25a shown in Fig.~\ref{fig:c-k-c}(b), and the MMPH master
sets considered in Sec.~\ref{subsec:crit}.
Other considered critical non-binary MMPHs satisfy only calibrated
inequalities but that is sufficient for experimental verification
of contextuality and possible applications.

We get thousands of MMPH non-binary sets as follows.
For the time being, we start with the previously
found KS sets: Bub 49-36, Conway-Kochen 51-37, Peres 57-40,
and Kochen-Specker's 192-118 which are all critical, i.e., if we
took out any edge from any of them they would stop being KS
\cite[Def.~3]{pavicic-pra-17}.
However, when we strip all the vertices contained in only one
edge we obtain Bub 33-36, Conway-Kochen 32-37, Peres 33-40, and
Kochen-Specker's 117-118 master sets, none of which is critical.
This enables us to generate thousands of new smaller MMPH critical
sets from them via our programs. Their distributions are shown in
Fig.~\ref{fig:dis-strip}. Chosen MMPHs critical sets are
given in Sec.~\ref{subsec:crit} and Appendix \ref{app:1} and shown in
Figs.~\ref{fig:bub-c} and \ref{fig:c-k-c}. They can be easily
implemented, in particular the smaller ones. 

The large number of obtained sets can also be used for an automated
testing of Yu-Oh's operators and inequalities along the examples we
gave in Secs.~\ref{subsec:yuohks} and \ref{subsec:crit}. For that we
are developing new algorithms and programs. This is a work in progress.

Next, one can make use of the obtained MMPHs to formulate new
entropic tests of contextualities following Kurzy{\'n}ski, Ramanathan,
and Kaszlikowski \cite{kurz-raman-12}. In 2012 they only had one
pentagram/pentagon set \cite{klyachko-08} at their disposal. The
pentagon 5-5 set is the simplest MMPH set we obtained
(see Fig.~\ref{fig:dis-strip}) and many other generated small sets
can now serve the purpose.

Also, the methods of evaluating conditions for being a SIC set
developed in \cite{ram-hor-14,cabell-klein-budr-prl-14} and the
methods of Cabello-Severini-Winter graph-theoretic approach to
quantum correlations \cite{cabello-severini-winter-14} require
samples of hypergraphs and that is what our method offers---a
constructive probabilistic generation of arbitrary MMPH sets when
coupled with automated vector generation algorithms we developed
in \cite{pm-entropy18}.

Finally, we stress that the MMPH constructive generation of
non-binary quantum sets from operationally chosen vectors out of
all possible ones within such sets contribute to our understanding
of the physical origin of quantum correlations since they represent
a new MMPH {\em scenario} for getting ``quantum correlations
from simple assumptions'' presented in \cite{cabello-19}.

\section{\label{sec:met}Methods}

The methods we use to handle quantum contextual sets rely on
algorithms and programs within the MMP language: 
VECFIND, STATES01, MMPSTRIP, MMPSHUFFLE, SUBGRAPH, LOOP, and
SHORTD developed in
\cite{bdm-ndm-mp-1,pmmm05a-corr,pm-ql-l-hql2,pmm-2-10,bdm-ndm-mp-fresl-jmp-10,mfwap-11,mp-nm-pka-mw-11,megill-pavicic-mipro-17}. They are
freely available at http://goo.gl/xbx8U2.  MMPHs can be
visualized via hypergraph figures consisting of dots and lines
and represented as a string of ASCII characters. The latter
representation enables processing billions of MMPHs simultaneously
via supercomputers and clusters. For the latter elaboration, we
developed other dynamical programs specific program to handle and
parallelize jobs with arbitrary number of MMP hypergraph vertices
and edges.

\vspace{6pt}

\begin{acknowledgments}
Supported by the Ministry of Science and Education
of Croatia through the Center of Excellence for Advanced
Materials and Sensing Devices (CEMS) funding, and by
MSE grants Nos. KK.01.1.1.01.0001 and 533-19-15-0022.
Computational support was provided by the cluster Isabella of
the Zagreb University Computing Centre and by the Croatian
National Grid Infrastructure (CRO-NGI).Technical supports by
Emir Imamagi\'c and Daniel Vr\v ci\'c from Isabella and
CRO-NGI are gratefully acknowledged.
\end{acknowledgments}

\begin{widetext}

\appendix
\section{\label{app:1} ASCII strings from MMPH non-binary
     classes}

   Below we give several chosen standard subgraphs from the four
   classes of critical MMPH sets shown in Fig.\ref{fig:dis-strip}.
   The first number in each line is $m$ of the biggest $m$-gon loop for
   the MMPH in the line. The second and third numbers are the numbers
   of the MMPH vertices and edges, respectively. Three commas
   ``{\tt ,,,}'' denote the end of a loop and * behind an ASCII
   symbol means that the symbol belongs to the loop.
   
\subsection{\label{app:1a} Bub's class}

10-v18-e13 \quad {\tt 213,36,6GC,CDB,BH8,89,9I4,45,5EA,A2,,,73*,9*2*,FD*7.}

11-v21-e16 \quad {\tt 213,3A,AHG,GFE,E57,76,6KL,LD8,89,9IC,C2,,,45*,B3*,D*2*,JF*B,H*8*4.}

14-v24-e18 \quad {\tt 12,2L3,34,4KG,GHI,I85,56,6B,BC,CA,A9,9FE,ED,DO1,,,78*,JH*F*,MN7,ND*C*.}

13-v27-e20 \ {\tt 213,3L4,45,5B,BC,CMN,NOE,E6F,FD9,9A,AJI,IHG,GP2,,,6*7,87,D*3*,KL*H*,QRO*,O*82*,RD*B*.}

17-v30-e23 \quad {\tt 543,3PC,CB,BA,AON,NJ6,67,7KL,L2S,SRI,IH,HTM,MGD,DE,E9,98,8Q5,,,12*3*,FG*,J*5*, P*O*M*,US*Q*,N*F2*.}

17-v33-e26 \quad {\tt 45,5CL,L7E,EF,FG,GBH,HIJ,JN,NRS,SWO,O6P,PMQ,QA2,2V8,89,9TU,U34,,,12*3*,6*7*,A*B*, C*D,KL*H*,G*D,M*J*,O*H*3*,XU*R*.}

\subsection{\label{app:1b} Conway-Kochen's class}
   
8-v15-e11 \quad {\tt 12,2E7,78,8D3,34,4C6,65,5F1,,,9AB,B7*6*,A3*1*.}

12-v22-e16 \quad {\tt 67,7GF,FB,B5D,D3,3ME,EC,CK8,89,9HI,I2A,AL6,,,12*3*,45*,A*B*C*,J42*.}

14-v26-e19 \quad {\tt 312,2F,FMN,NL5,596,67,7OJ,JIE,EB,BA,APD,DC,CQH,H3,,,45*3*,89*,G2*, KL*G,I*H*8.}

15-v29-e22 \quad {\tt 12,2RH,HQ3,34,47M,MT9,9A,AJE,ED,DIF,FG,GC,CB,B5N,NS1,,,5*6,7*8,H*I*,KL8,LD*6,OPG*, PN*M*.}

17-v30-e24 \ \ \ {\tt 12,2TD,DH,HRO,O87,76,65,5P4,43,3SJ,JK,K9L,LIM,MQN,NCE,EB,BU1,,,8*9*,AB*,C*D*,FG,I*G, P*GA,Q*J*F.}

\subsection{\label{app:1c} Peres' class}

10-v15-e12 \quad {\tt 12,2A,AC8,87,7D5,56,6B9,94,43,3E1,,,E*C*B*,FE*D*.}

14-v19-e16 \quad {\tt 12,23,34,4E,EGA,A9,9HB,BC,CFD,D8,87,7I6,65,51,,,I*G*F*,JI*H*.}

14-v27-e19 \quad {\tt 12,2QD,DE,E3I,IJK,KM5,56,6L8,87,7PG,GHF,FAB,BC,CR1,,,3*4,9A*,E*C*,NJ*9, OH*4.}

20-v35-e27 \quad {\tt 213,3G,GLM,MNE,EF,FVX,XYU,UP5,5I,IT7,78,89,9S6,6J,JZQ,QHA,AB,BKD,DC,CR2,,,45*6*, H*I*,F*3*,OP*K*,V*Q*L*,T*S*2*,WX*R*.}

22-v38-e30 \quad {\tt 345,5SU,UTH,HI,IR2,2cZ,ZFa,aJW,WVX,XQG,G7,76,6LC,CB,BMD,DE,EYA,A9,98,8ON,NbK,KP3,,, 12*3*,F*G*,J*5*,J*I*,K*E*,P*Q*L*,R*S*M*,a*Y*O*.}

\subsection{\label{app:1d} Kochen-Specker' class}

7-v12-e9 \quad {\tt 12,23,34,456,6A9,987,7C1,,,5*1*,B8*3*.}

12-v19-e14 \quad {\tt 12,2IA,AB,BC8,87,7E5,56,6D4,43,3FG,G9H,HJ1,,,9*A*,H*D*C*.}

16-v30-e21 \quad {\tt 312,2E,EMN,NL8,8RC,CD,D7,76,6GH,HP9,9SA,AB,BQ4,45,5IJ,JO3,,,8*9*3*,F2*,KL*F,TO*B*,  UP*D*.}

18-v38-e27 \quad {\tt 34,4VD,DE,ETG,GF,FcO,ON,NHJ,JK,KSC,CB,BZ5,56,6Y9,9A,AX7,78,8W3,,,12,H*I,LM,PQ,RQ, UMI,aR2,bP1,QN*L.}

18-v46-e33 \quad {\tt 56,6a8,87,7e9,9A,AcC,CB,BdD,DE,EbG,GF,FiP,PQ,QYX,XT,THJ,JK,Kh5,,,12,34,H*I,LM,NO, RS,T*US,VU,WU,ZRO,fI4,gM3,jW2,kV1,T*NL.}

12-v54-e39 \quad {\tt 78,8oV,VW,WgX,XY,YZ,ZfU,UT,TpA,A9,9Ps,sN7,,,12,34,56,BC,DE,FG,HI,JK,LM,N*O,P*Q,RS, ab,cb,dcS,eaR,hK2,iI1,jM6,kOG,lQ5,mC4,nE3,qY*D,rbB,s*LH,s*JF.}

\end{widetext}

\end{document}